\documentclass[a4paper]{article}

\usepackage[T1]{fontenc}
\usepackage{graphicx}
\usepackage{hyperref}
\usepackage{amsmath}
\usepackage[capitalize]{cleveref}
\usepackage{color}
\usepackage{fullpage}

\usepackage{optidef}
\usepackage{amssymb,amsthm}

\theoremstyle{plain}
\newtheorem{theorem}{Theorem}
\newtheorem{lemma}{Lemma}

\theoremstyle{definition}
\newtheorem{definition}{Definition}

\usepackage{algorithm}
\usepackage[noend]{algpseudocode}
\usepackage{subcaption}

\newcommand{\appfac}[0]{\ensuremath{2.097}\xspace}
\newcommand{\appfackpath}[0]{\ensuremath{2.41}\xspace}
\newcommand{\appfacpcsymbol}[0]{{\hat{\alpha}}}
\newcommand{\solbeforepc}[0]{{T^{\prime\prime}}}
\newcommand{\pickupforest}[0]{{F_P}}
\newcommand{\sigzero}[0]{{\sigma_0}}
\newcommand{\sigzeroprime}[0]{{\sigma'_0}}

\newcommand{\mystar}{*}
\usepackage{xspace}
\newcommand{\OTSP}{OTSP\xspace}
\newcommand{\PCOTSP}{PCOTSP\xspace}
\newcommand{\kTSP}{Multi-Path-TSP\xspace}
\newcommand{\kpath}{Multi-Path\xspace}
\newcommand{\PCkTSP}{PC-Multi-Path-TSP\xspace}
\newcommand{\PCTSP}{PCTSP\xspace}
\newcommand{\expectation}[1]{{\mathbb E \bigl[ {#1} \bigr]}}
\renewcommand{\Pr}{\text{Pr}}

\usepackage{subcaption}
\usepackage{authblk}
\usepackage{orcidlink}

\title{Approximating Prize-Collecting Variants of TSP}
\author{Morteza Alimi \orcidlink{0000-0003-3732-2903}}
\author{Tobias M\"{o}mke \orcidlink{0000-0002-2509-6972} \thanks{Partially supported by DFG Grant 439522729 (Heisenberg-Grant).}}
\author{Michael Ruderer \orcidlink{0009-0009-1968-4821}}
\affil{University of Augsburg}
\begin{document}
\maketitle   
\begin{abstract}
    We present an approximation algorithm for the 
    \textit{Prize-col\-lect\-ing Ordered Traveling Salesman Problem} (\PCOTSP),
    which simultaneously generalizes the Prize-collecting TSP and the Ordered TSP.
    The Prize-collecting TSP is well-studied and has a long history, 
    with the current best approximation factor slightly below $1.6$, shown by Blauth, Klein and Nägele [IPCO 2024].
    The best approximation ratio for Ordered TSP is $\frac{3}{2}+\frac{1}{e}$, 
    presented by  B\"{o}hm, Friggstad, M\"{o}mke, Spoerhase [SODA 2025] and Armbruster, Mnich, N\"{a}gele [Approx 2024]. 
    The former also present a factor 2.2131 approximation algorithm for \kTSP.

    By carefully tuning the techniques of the latest results on the aforementioned problems 
    and leveraging the unique properties of our problem, 
    we present a \appfac-approximation algorithm for \PCOTSP. 
    A key idea in our result is to first sample a set of trees, and then
    probabilistically pick up some vertices,
    while using the pruning ideas of Blauth, Klein, N\"{a}gele [IPCO 2024] on other vertices
    to get cheaper parity correction; 
    the sampling probability and the penalty paid by the LP playing a crucial part in both cases.
    A straightforward adaptation of the aforementioned pruning ideas would only give minuscule improvements over standard
    parity correction methods.
    Instead, we use
    the specific characteristics of our problem together with properties gained from running a simple combinatorial algorithm
    to bring the approximation factor below 2.1.

    Our techniques extend to Prize-collecting \kpath TSP, building on results from
    B\"{o}hm, Friggstad, M\"{o}mke, Spoerhase [SODA 2025], leading to a \appfackpath-approximation.
\end{abstract}

\section{Introduction}

The metric traveling salesman problem (TSP), which asks for a shortest closed tour in a 
metric space $(V,c), \ c: V\times V \to \mathbb{Q}_+$ on an $n$ element vertex set $V$ 
visiting each vertex exactly once%
\footnote{Note that the problem can alternatively be defined as having a weighted graph as input, and seeking to find 
a shortest closed walk (or shortest Eulerian multi-subgraph) that visits each vertex at least once.}
is one of the most well-studied problems 
in combinatorial optimization in its various incarnations. 
Christofides \cite{christofides-1976} and Serdjukov \cite{serdjukov-1978}  
gave a simple  $\frac{3}{2}$ approximation algorithm
for  symmetric (undirected) TSP;  
an approximation factor slightly below $\frac{3}{2}$ was provided by Karlin, Klein and Oveis Gharan~\cite{karlin-klein-oveis-2021}.

Also the path versions of TSP have been subject to intense study. This line of study led to a surprising outcome:
Traub, Vygen, and Zenklusen \cite{traub-vygen-zenklusen-2020} give a reduction from path-TSP to TSP 
which only loses $\epsilon$ in the approximation factor.  
In their book, Traub and Vygen \cite{traub-vygen-2024} 
simplify the reduction and use \kTSP (see below) as a building block.
The book gives a comprehensive overview of the aforementioned results and more, including 
asymmetric (directed) variants of TSP.

A more general, and more practical version of the problem can be defined 
by allowing the tour to omit some cities by paying some additional penalty. 
This \textit{prize collecting} paradigm has been intensely studied for various combinatorial optimization problems; 
see, e.g., \cite{ahmadi-gholami-hajiaghayi-2024a,ahmadi-gholami-hajiaghayi-2024b} for some recent results 
regarding Prize Collecting Steiner Forest and Steiner Tree Problems. 
As regards the \text{Prize-collecting Traveling Salesman Problem (\PCTSP)}, 
Bienstock et.\ al.\ \cite{bienstock-goemans-simchilevi-1993} give a $2.5$-approximation for this problem. 
Goemans and Williamson \cite{goemans-williamson-1995} give a factor 2 approximation. 
The first group to break the 2 barrier were Bateni et.\ al.\ \cite{archer-bateni-hajiaghayi-2011}, 
who achieved a factor of $1.979$. 
Goemans \cite{goemans-2009} observed that by carefully combining the two aforementioned algorithms, one can achieve 
an approximation factor of $1.91$. 
Blauth and N\"{a}gele \cite{blauth-naegele-2023} gave a factor $1.774$-approximation. 
Blauth, Klein, and N\"{a}gele~\cite{blauth-klein-naegele-2024} achieved a factor of $1.599$.

Another important generalization of TSP is attained by enforcing some ordering on a subset of vertices. 
Formally, the metric Ordered Traveling Salesman Problem (\OTSP), we are given 
a metric 
graph $G=(V,E),c\colon E\to \mathbb{R}_+$ together with $k$ terminals $O=\{o_1, \ldots, o_k\}$, 
and the objective is to find a shortest tour that visits the terminals in order. 
Any $\alpha$-approximation for TSP can be utilized to give an $\alpha+1$-approximation algorithm for \OTSP
by finding an $\alpha$-approximate tour on $V\setminus O$ and combining it with the cycle on $O$ \cite{boeckenhauer-hromkovic-kneis-2006}.
Furthermore, there is a combinatorial $(2.5-2/k)$-approximation algorithm~\cite{boeckenhauer-moemke-steinova-2013}.
A substantial improvement in the approximation factor to $\frac{3}{2} + \frac{1}{e}$ was achieved by
\cite{boehm-friggstad-moemke-2025} and \cite{armbruster-mnich-nagele-2024}. 

In \kpath Traveling Salesman Problem (\kTSP), in addition to a weighted graph we are given a list of $2k$ terminals $\mathcal{T} = \{(s_1,t_1),\ldots,(s_k,t_k)\}$.
The objective is to find $k$ paths $P_i, \ 1\le i\le k$ 
of minimum total length, partitioning the vertex set of the graph (see also Section 16.4 of ~\cite{traub-vygen-2024}). 
Böhm, Friggstad, Mömke and Spoerhase~\cite{boehm-friggstad-moemke-2025} give a factor $2.2131$ approximation algorithm for the problem.

\subsection{Our Contribution and Overview of Techniques}
We study \OTSP and \kTSP in the prize-collecting setting, which can also be viewed as generalizations of 
\PCTSP.

\begin{definition}
    In the \textit{Prize-collecting Ordered Traveling Salesman Problem (PC\-OTSP)}  
    we are given a metric weighted complete graph with penalties for vertices  
    $G=(V,E), \ c\colon E\to \mathbb{R}_+, \pi:V\to \mathbb{R}_+$, 
    together with $k$ terminals $O=\{o_1, \ldots, o_k\}$,
    and the objective is to find a tour $C$, minimizing $c(C) + \sum_{v\not\in V(C)} \pi(v)$. 
\end{definition}
\begin{definition}
    In the \textit{Prize-collecting \kpath Traveling Salesman Problem (\PCkTSP)}, 
    we are given a metric weighted complete graph with vertex penalties  
    $G=(V,E), \ c\colon E\to \mathbb{R}_+, \pi\colon V\to \mathbb{R}_+$, 
    together with a set of $k$ terminal pairs $\mathcal{T}=\{(s_1,t_1), \ldots, (s_k,t_k)\}$,
    and the objective is to find $k$ paths $P_i, 1\le i\le k$, such that $P_i$ is a path with endpoints $s_i, t_i$,
    minimizing $\sum_{i=1}^k c(P_i) + \sum_{v\not\in (\cup V(P_i))} \pi(v)$. 
\end{definition}

\begin{theorem}\label{thm:opctsp}
    There is a \appfac-approximation algorithm for \PCOTSP. 
\end{theorem}

\begin{theorem}\label{thm:kpathtsp}
    There is a \appfackpath-approximation algorithm for \PCkTSP. 
\end{theorem}

In the remainder of this section we provide an overview of our algorithms and techniques for \PCOTSP and \PCkTSP. 
The details of the algorithms and the proofs of correctness are covered in the following sections

In \cref{sec:pcotsp}, we specify the LP relaxation (OLP) of \PCOTSP.
Our algorithm first solves the LP, and then it samples a tree $T_i$ for each of the $k$ parts of the 
LP solution, in the manner of Lemma 2 in \cite{blauth-klein-naegele-2024}, 
and Theorem 4 in \cite{boehm-friggstad-moemke-2025},  
which are both based on a result of Bang Jensen, Frank and Jackson~\cite{bangjensen-frank-jackson-1995} 
and Post and Swamy~\cite{post-swamy-2015}\footnote{Post and Swamy showed that the Bang Jensen, Frank, Jackson decomposition can be computed in polynomial time.}
(see \cref{lem:bang-jensen}). 
The expected cost of each tree is at most the corresponding part of the LP solution, and its expected coverage
is at least that of the corresponding $y$ values.
Now a key idea in our algorithm is to \textit{even out} the penalty ratios paid by the algorithm, 
by either reducing the penalties paid by some vertices through connecting them up with the $T_i$'s 
(which inevitably increases the connectivity cost), 
or utilizing the spare penalty ratios of other vertices to bring down the cost of parity correction. 
More precisely, we use the \textit{pruning} idea from \cite{blauth-klein-naegele-2024} to probabilistically
prune away portions of the resultant structure with low connectivity. 
This allows us to assign weights to the edges of the tree, which, when combined with a suitable multiple of 
the LP solution, gives a point in the $Q$-join polytope, where $Q$ is the set of odd degree vertices of the current structure. 
This allows for cheaper parity correction, because the tree edges with higher assigned weights, 
which are associated with lower value cuts, are pruned with higher probability, 
and thus the expected cost incurred by them remains low. 

The prominent difference between our problem and the standard \PCTSP, which necessitates finding new ideas for pruning, 
is the sampling probability for vertices.
In \PCTSP, the sampling probability for each vertex $v$ is equal to $y_v$, hence the expected penalty ratio
for $v$ is one, and intuitively, there is lots of spare penalty for $v$ to utilize for pruning. 
But in \PCOTSP, the sampling probability of a vertex $v$ is lower (bounded from below by $1-e^{-y_v}$); 
hence there is little spare penalty ratio left (and for fewer vertices) to utilize for pruning. 
Thus, in order to get a nontrivial improvement from pruning we need new ideas. 

This can be partly achieved by combining our algorithm with a version of the classical algorithm for \OTSP~\cite{boeckenhauer-hromkovic-kneis-2006}. 
That algorithm simply combines the cycle on terminals with a $\frac{3}{2}$-approximation for TSP on the remaining vertices.
If the length of the cycle on terminals is low, this algorithm gives a good overall approximation ratio.
If the length of the cycle is large, this fact can be utilized to improve the analysis for parity correction, because we can show that there is no need to assign (nonzero) weights to the cycle edges. 

We believe that this view of \textit{evening out} the penalty ratio of vertices vis-a-vis the optimal LP
solution is a very useful conceptual tool for approaching prize collecting versions of TSP 
or other combinatorial optimization problems.
In the case of \PCOTSP, it is partly realized through the new idea of probabilistically including in the solution 
those vertices which currently have a too high penalty ratio (based on the desired approximation factor).
Squeezing out the slack penalty ratio (again, based on the target approximation factor), 
is achieved through pruning. 
Here, the specific structure of the \OTSP, and the partially constructed solution proves useful; 
we can show that if the current cycle through the terminals is small, 
combining the cycle on terminals with the best \PCTSP approximation for the rest of vertices gives a good factor, 
while a large cycle improves the cost of pruning.

In the case of \kTSP, B\"{o}hm, Friggstad, M\"{o}mke, and Spoerhase \cite{boehm-friggstad-moemke-2025} 
propose two algorithms, and show that a careful combination of them leads to a good approximation. 
The first one, which doubles all edges except those on $s_i,t_i$ paths 
(and is good when the sum of $s_i,t_i$ distances is large) 
can be adapted to our setting by probabilistically picking up vertices with high penalty ratio, 
as opposed to picking up all remaining vertices in \cite{boehm-friggstad-moemke-2025}. 
We replace their second algorithm, which finds a minimum length forest in which each terminal appears 
in exactly one of the components and then adds direct $s_i,t_i$ edges with an algorithm that
uses \cref{lem:bang-jensen} to sample a tree from the related \PCTSP, 
and then adds direct $s_i,t_i$ edges. 

\section{Prize-collecting Ordered TSP}\label{sec:pcotsp}
In this section we describe our algorithm for \PCOTSP. 
Some of the technical proofs will be presented in the following sections. 

\subsection{Preliminaries and Definitions}
In the \OTSP, a tour can be decomposed into $k$ subtours, between $o_i$ and $o_{i+1}$.\footnote{For notational convenience, we identify $o_{k+1}$ with $o_1$.} 
We can take the polyhedron determined by the following inequalities as the relaxation of one such subtour between two vertices $s$ and $t$.
For each vertex $v$, the variable $y_v$ indicates its fractional degree in the stroll.
\begin{equation*}
    \begin{array}{lrlll}
        & y_{s} = y_{t} &= \frac 1 2 & \\[5pt]
        & x(\delta(v)) &= 2 y_{v} & \forall v \in V \\[5pt]
        & x(\delta(S)) &\geq 1 \quad  & \forall S \subseteq V \setminus \{t\}, s \in S \\[5pt]
        & x(\delta(S)) &\geq 2 y_v \quad  & \forall v \in S \subseteq V \setminus \{s, t\} \\[5pt]
        &   x, y   &\geq 0 &
    \end{array}
    \quad\quad
    \textbf{($s$-$t$-stroll relaxation)}
\end{equation*}

Now, the \PCOTSP can be modeled as the following linear program. 
For $i = 1, \ldots k$ we define $x_i$ and $y_i$ to be the vectors $(x_{i,e})_{e \in E}$ and $(y_{i,v})_{v \in V}$, respectively.
For each $i$, the vector $(x_i, y_i)$ is constrained to be a feasible point in the $o_i$-$o_{i+1}$-stroll relaxation, 
the sum $y_v = \sum_{i=1}^k y_{i, v}$ over all fractional degrees of a vertex $v$ is an indicator of to what degree $v$ is used in the solution. 
When $v$ is not fully used, i.e., when $y_v < 1$, the LP has to pay a fractional penalty of $\pi_v (1 - y_v)$.
We will usually refer to the pair $(o_i, o_{i+1})$ as $(s_i, t_i)$, to emphasize that $(x_i, y_i)$ is a fractional tour/stroll from $o_i$ to $o_{i+1}$.
\begin{mini*}
    {}{\displaystyle\sum\limits_{e \in E} \displaystyle\sum\limits_{i = 1}^k c_e x_{i, e} + \displaystyle\sum\limits_{v \in V} \pi_v (1 - y_v)}{\label{lp:pcotsp}}{}
    \addConstraint{y_v}{= \sum\limits_{i = 1}^k y_{i, v}\quad }{\forall v \in V} 
    \addConstraint{(x_i, y_i) \text{ lies in the }s_i\text{-}t_i\text{-stroll relaxation}}{}{}
\end{mini*}

We refer to this LP as (OLP).
Note that every feasible solution to (OLP) has $y_o = 1$ for all terminals $o$.
Similar to $y_v$, we will use $x_e$ as a shorthand for $\sum_{i=1}^{k} x_{i, e}$.
It follows immediately from the LP constraints that $x(\delta(S)) \geq 2y_v$ for any $v \in S \subseteq V \setminus O$. 
By contracting all terminals $o \in O$ into a single vertex $r$, we therefore obtain the relaxation for 
the normal \PCTSP from \cite{blauth-klein-naegele-2024}, 
which also involves a root, which (without loss of generality; see, e.g., \cite{archer-bateni-hajiaghayi-2011}) is required to be in the tour.
Given a solution $(x, y)$ to (OLP) and some threshold $\rho \in [0, 1]$, we define $V_\rho \coloneqq \{v \in V \mid y_v \geq \rho\}$.

\subsection{A simple algorithm}
We consider the following simple algorithm for \PCOTSP, inspired by~\cite{boeckenhauer-hromkovic-kneis-2006}.
First, directly connect the terminals $o_1, \ldots, o_k$ in order, creating a simple cycle $\hat C$.
Then, compute a solution to the \PCTSP on the same instance. 
Since all terminals have infinite penalty, every tour $T$ obtained in this way includes all terminals.
We obtain an ordered tour of cost $c(T) + c(\hat C)$ by following the original cycle $\hat C$ and grafting in the tour at an arbitrary terminal $o \in O$.

Using an $\appfacpcsymbol$-approximation algorithm to solve the \PCTSP, we know that the sum of tour- and penalty costs for this solution is at most $\appfacpcsymbol \cdot \text{opt}_\text{\PCTSP} \leq \appfacpcsymbol \cdot \text{opt}_\text{\PCOTSP}$.
Since $c(\hat C) \leq \text{opt}_\text{\PCOTSP}$, this immediately implies a $(1 + \appfacpcsymbol)$-approximation algorithm for \PCOTSP.
One can see that this algorithm performs even better if we can guarantee that $\hat{C}$ is small. 
To be precise, for any $\alpha \geq \appfacpcsymbol$ we obtain an $\alpha$-approximation provided that $c(\hat C) \leq (\alpha - \appfacpcsymbol) \text{opt}_\text{\PCOTSP}$.
We therefore may assume $c(\hat C) \geq (\alpha - \appfacpcsymbol) \text{opt}_\text{\PCOTSP}$ in the analysis of our main algorithm.
The currently best value for $\appfacpcsymbol$ is the approximation guarantee of (roughly) 1.599 obtained by \cite{blauth-klein-naegele-2024}.

\subsection{Overview of our main algorithm}
Fix $\alpha=\appfac$. 
We first solve (OLP)%
\footnote{The separation oracle for the LP boils down to separating subtour elimination constraints. Hence the LP can be solved in polynomial time using 
the ellipsoid algorithm.},
to get an optimal solution $(x^\ast,y^\ast)$. 
Using the following \cref{lem:splitting_off}, we then \emph{split off} the vertices $v$ 
for which $y_v \le \theta$, for a parameter $\theta$ to be determined later, 
to get a solution $(\hat{x}, \hat{y})$ for the LP were all vertices have a certain minimum connectivity to the terminals.
Our solution will not include these vertices in the final tour, and instead pays the full penalty for them. We state the following lemma for \PCOTSP, as the proof is identical to the one for \PCTSP.

\begin{lemma}[Splitting off~\cite{blauth-klein-naegele-2024}]\label{lem:splitting_off}
    Let $(x, y)$ be a feasible solution to the \PCOTSP relaxation and let $S \subseteq V \setminus O$.
    Then we can efficiently compute another feasible solution $(x^\prime, y^\prime)$ in which $y^\prime_v = 0$ for all $v \in S$, but $y'_v = y_v$ for all $v \not \in S$, and $c(x^\prime) \leq c(x)$.
\end{lemma}

Lemma \ref{lem:splitting_off} ensures that we can remove vertices from the support of $x$, without increasing the cost of $x$. Since we only split off vertices for which $y^\ast_v$ is relatively low, we can afford to pay the additional penalties if we set $\theta = 1 - \frac{1}{\alpha}$ (we will prove this in \cref{sec:penalty-cost}).
We proceed by sampling a set of trees based on $(\hat x, \hat y)$. 

\begin{lemma}\label{lem:bang-jensen}(\cite{blauth-klein-naegele-2024}, following \cite{bangjensen-frank-jackson-1995})
    Suppose $(x, y)$ is a feasible point in the $s$-$t$-stroll relaxation. 
    In polynomial time we can find a set of trees $\mathcal{T}$ and weights $(\mu(T))_{T \in \mathcal T}$
    such that (i) 
    $\sum_{T \in \mathcal T} \mu(T) = 1$; and (ii)
    $\sum_{T\in \mathcal T : e \in E(T)} \mu(T) \leq x_e$; and (iii)
    $\sum_{T\in \mathcal T : v \in V(T)} \mu(T) \geq y_v$; and (iv)
    all trees span $s$ and $t$.
\end{lemma}
Lemma \ref{lem:bang-jensen} can be used to sample a random tree of expected cost $c(x)$ which contains each vertex $v$ with probability at least $y_v$.
It is straightforward to see that this can be achieved by choosing each tree $T \in \mathcal T$ with probability $\mu(T)$.

For each component $(\hat{x}_i, \hat{y}_i)$ of our solution, we apply Lemma~\ref{lem:bang-jensen} and sample a tree $T_i$ as we have described.
Define $T:= {\dot{\bigcup}}_{i=1}^{k}T_i$.
Note that $T$ is no longer a tree, and it might even have repeated edges. 

\subsubsection{Picking up and Pruning}
The probability that any vertex $v$ is not in $T$ is at most $e^{-y^\ast_v}$ (we will formally prove this in \cref{lem:sampling-prob}), 
while the penalty that the LP pays for $v$ is $1-y^\ast_v$. 
We define the \emph{penalty ratio} $\rho_v$ of a vertex $v$ as $\rho_v = \frac{e^{-y^\ast_v}}{1-y^\ast_v}$. 
Since we want to obtain an $\alpha$-approximation, we define $\sigma_0$ as the value for which this ratio is equal to $\alpha$, i.e., $\sigma_0$ is the unique solution to the equation $\alpha (1 - \sigzero) = e^{-\sigzero}$.

As $\rho_v$ is an increasing function of $y^\ast_v$, vertices $v$ for which $y_v > \sigma_0$ have a penalty ratio greater than $\alpha$.
In fact, their penalty ratio becomes arbitrarily high as $y^\ast_v$ approaches $1$. 
We define $V_{\sigma_0} = \{v \in V \mid y^\ast_v \geq {\sigma_0}\}$ to be \emph{critical} vertices.
Intuitively, this means that the probability of these vertices to not be sampled is too high.
We \emph{pick up} these critical vertices with certain probability.
That means, we will probabilistically select some unsampled critical vertices and connect them to our solution.
We first describe how these vertices are selected and then give the details on how we connect them.

Note that by our observation on $\rho_v$, the chance for picking up a fixed unsampled critical vertex $v$ should increase with the value of $y^\mystar_v$.
Our strategy is therefore the following: we first draw a global threshold $\sigma \in [\sigma_0, 1]$ from a suitable distribution $D_\sigma$, and then pickup all critical vertices whose $y$-value is above $\sigma$, i.e., all vertices in $V_{\sigma} \setminus V(T)$.

The distribution $D_\sigma$ is determined by the cumulative probability function $F_\sigma$ which we state in the following.
Choosing $\sigma$ according to this function will ensure that for any $v \in V_{\sigma_0}$ the probability of being picked up is just high enough so that the expected penalty paid for $v$ is at most $\alpha$ times the fractional penalty $\pi_v(1 - y^\ast)$ paid by the LP solution. We will formally prove this fact in \cref{sec:penalty-cost}.
For now, we simply define $F_\sigma(y) \coloneqq P\bigl[\sigma \leq y\bigr] \coloneqq 1 -  \frac{\alpha(1 -y)}{e^{-y}}$.
We remark that by definition of $\sigzero$, we have $\frac{\alpha(1 -y)}{e^{-y}} = \frac{e^y(1 -y)}{e^\sigzero (1-\sigzero)}$, from which it is easy to verify that $F_\sigma$ is indeed a cumulative probability function, i.e., $F_\sigma(\sigzero) = 0$, $F_\sigma(1) = 1$ and $F_\sigma$ is monotonously increasing on $[\sigzero, 1]$.

To properly describe how we connect $V_\sigma \setminus T$ to $T$, we use the following definition from \cite{boehm-friggstad-moemke-2025}.
\begin{definition} \label{def:rooted-spanning-forest}
    Let $X \subseteq Y \subseteq V$. An $X$-rooted spanning forest of $Y$ is a spanning forest of $Y$ such that each of its connected components contains a vertex of $X$.
\end{definition}
A cheapest $X$-rooted spanning forest of $Y$ can be efficiently computed by contracting $X$, computing a minimal spanning tree of $Y$ in the contracted graph and then reversing the contraction.

For our pickup step, we buy the cheapest forest $\pickupforest$ that spans $V_\sigma$ and is rooted in ${V_\sigma \cap V(T)}$. 
This forest $\pickupforest$ has two crucial properties: (i) after buying $\pickupforest$, each vertex $v \in V_\sigma \setminus T$ will be connected to $T$; and (ii) $\pickupforest$ only spans vertices with $y^\ast_v \geq \sigma$.
Property (ii) follows immediately from the definition of $\pickupforest$, and property (i) follows from the fact that $O \subseteq V_\sigma \cap T$ for any choice of $\sigma$.

We now move on to those vertices for which $y^\ast_v < \sigma_0$. These vertices $v$ have some \textit{spare} penalty ratio, which we utilize to prune $v$ 
with certain probability. 

Note that in the \PCTSP regime, e.g., in the result of \cite{blauth-klein-naegele-2024}, 
the probability that a vertex $v$ is in the (one) sampled tree is exactly $y^\ast_v$, which means that $\rho_v=1$ for every $v$. 
But in our setting, the probability of a vertex not being sampled is [bounded by] $e^{-y^\ast_v}$, 
which intuitively means that the spare penalties are much more constrained, and the grains would be more meager. 
Nonetheless, we show how to utilize the specific structure of our problem to gain almost as much from pruning as 
in the setting of \PCTSP. 

\begin{algorithm}
    \caption{The Prize-collecting Ordered TSP Algorithm \label{alg:pcotsp}}
    \begin{algorithmic}[1]
        \Statex \textbf{Input:} An \PCOTSP instance
        \State Compute an optimal solution $(x^\ast, y^\ast)$ of OLP
        \State Split off vertices $v \in V$ with $y^\ast_v \leq \theta$ \label{line:splitting-off}
        \State Sample $k$ trees $T_1, \ldots T_k$
        \State Sample a pruning threshold $\gamma \sim D_\gamma$
        \State Prune each $T_i$ based on $\gamma$ to get $T'_i$. Let $T'=\bigcup_{i=1}^k T'_i$
        \State Sample a pickup threshold $\sigma \sim D_\sigma$
        \State Pick up all unsampled vertices with $y_v \geq \sigma$ to obtain $T''$ 
        \State Add a shortest $\text{odd}(T'')$-$\text{join}$ $J$ to $T''$ to get $H$ 
        \State \Return an ordered tour obtained by shortcutting $H$ 
    \end{algorithmic}
\end{algorithm}

For our pruning step, we need the following definition from \cite{blauth-klein-naegele-2024}:
\begin{definition} For a fixed LP solution $(x, y)$, a tree $T$, and a threshold $\gamma \in [0, 1]$, we define $\text{core}(T, \gamma)$ as the inclusion-wise minimal subtree of $T$ that spans all vertices in $V(T) \cap V_\gamma$.
\end{definition}
Algorithmically, $\text{core}(T, \gamma)$ can be obtained by iteratively removing leaves $v$ with $y_v < \gamma$.
To prune a tree $T_i$ simply means to replace it by $T_i^\prime = \text{core}(T_i, \gamma)$.
We emphasize that for pruning $T_i$, we do not consider the local penalty values $y_{i, v}$, but the \emph{global} values $y_{v}$. 
Furthermore, our construction ensures that $s_i$ as well as $t_i$ are part of the pruned tree $T_i^\prime$ (remember that $y^\ast_{s_i} = y^\ast_{t_i} = 1$).
As we did for our pickup threshold, we will draw the global pruning threshold $\gamma$ according to some distribution $D_\gamma$ which is defined by the following cumulative probability function:
$F_\gamma(y) \coloneqq \frac{1 - \alpha(1 -y)}{1 - e^{-y}}$.
We will prove in \cref{sec:penalty-cost} that this ensures that the slack in the penalty ratios is fully utilized. 
Finally, we remark that our pickup and pruning steps target the disjoint vertex sets $\{v \in V \mid y^\ast_v \geq \sigma_0\}$ and $\{v \in V \mid \theta \leq y^\ast_v < \sigma_0\}$, so neither step interferes with the other.

\subsubsection{Obtaining an ordered tour}
To turn $T''$ into a feasible tour, we first need to correct parities, i.e., we ensure that every vertex has an even degree. Let $H$ be the graph that is obtained by adding a cheapest $\text{odd}(T'')$-join to $T''$, where $\text{odd}(T'')$ is defined as the set of odd degree vertices of $T''$. 
Observe that initially, each tree $T_i$ contains a path $P_i$ between $o_i$ and $o_{i+1}$. Since all terminals $o \in O$ have $y^\ast_o = 1$, all of these paths survive the pruning step.
So the multigraph $H$ still contains $k$ edge-disjoint $o_i$--$o_{i+1}$--paths which, taken together, form a closed walk $C$.
By removing $C$ from $H$, we obtain a graph whose connected components have even degree and can thus be shortcut into cycles. 
Furthermore, since $H$ was connected, each of these cycles has a common vertex with $C$.
We obtain an ordered tour through $H$ by following the walk $C$ and grafting in the cycles obtained by shortcutting the components of $H \setminus C$ on the way.

\section{Analysis of Algorithm \ref{alg:pcotsp}}

In this section, we prove \cref{thm:opctsp}. 
We will compare the expected cost of our computed solution to the cost of the optimal LP solution $(x^\ast, y^\ast)$.
In particular, we compare the expected cost of our computed tour to the cost of $x^\ast$ and the expected sum of our incurred penalties to the fractional penalty cost defined by $y^\mystar$. In the following, we use $(x^\ast, y^\ast)$ to refer to the optimal LP solution computed by the algorithm, and $(x, y)$ to refer to the LP solution after the splitting-off step.

\subsection{Bounding the Penalty Ratios} \label{sec:penalty-cost}
In this subsection, we prove the following lemma:

\begin{lemma}\label{lem:penalty-bound}
    The expected total penalty cost paid by Algorithm \ref{alg:pcotsp} is at most
    \[
        {\alpha \cdot \sum_{v \in V} \pi_v \cdot ( 1- y^\mystar_v)}.
    \]
\end{lemma}

\newcommand{\penaltyprob}[0]{{\Pr\bigl[v \not\in V(T'')\bigr]}}
We note that we can express the expected total penalty cost paid by our algorithm as 
\[
    \sum_{v \in V} \pi_v \cdot \penaltyprob.
\]
We can thus prove \cref{lem:penalty-bound} by showing that for each vertex $v$, the ratio $\frac{\penaltyprob}{1 - y^\mystar_v}$ is at most $\alpha$.

Due Step~\ref{line:splitting-off} in \cref{alg:pcotsp}, i.e., due to the splitting-off step, $T''$ spans only vertices whose {$y$-value} is at least $\theta$.
Vertices $v$ with $y^\mystar_v < \theta$ are therefore included in $V(T'')$ with probability~0.
However, our choice of $\theta = 1 - \frac{1}{\alpha}$ guarantees that for the vertices $v$ with $y^\mystar_v < \theta$,
\[
    \frac{\penaltyprob}{1 - y^\mystar_v} \leq \frac{1}{1 - \theta} = \alpha .
\]
To continue our analysis for the remaining vertices, i.e., those with $y^\mystar_v \geq \theta$, we show the following lemma:
\begin{lemma}\label{lem:sampling-prob}
    Let $v$ be a vertex with $y^\mystar_v \geq \theta$. 
    Then the probability that $v$ is \emph{not} in $T$ is at most $e^{-y^\ast_v}$
\end{lemma}
\begin{proof}
    By \cref{lem:bang-jensen}, the probability that $v$ is not in $T_i$ is at most $1 - y^\ast_v$ for any fixed $i \in [k]$. The probability that this happens for \emph{all} $i \in [k]$ is at most 
    \[
    \prod_{i = 1}^k \Pr[v \not \in V(T_i)] \leq \prod_{i=1}^k (1 - y^\ast_v) \leq \exp\left(-\sum_{i=1}^k y^\ast_{i, v}\right) = e^{-y^\ast_v}.
    \]
\end{proof}
Note that the distributions from which $\sigma$ and $\gamma$ are chosen guarantee that $\theta \leq \gamma < \sigzero \leq \sigma \leq 1$, i.e., only vertices $v$ with $y^\mystar_v \in [\theta, \sigzero)$ can be pruned and only vertices with $y^\ast_v \in [\sigzero, 1]$ can be picked up.

Now consider a vertex $v$ with $y_v \in [\theta, \sigzero)$. 
There are two cases in which we have to pay $v$'s penalty: If $v$ is not sampled, and if $v$ is sampled but immediately pruned afterwards.

By \cref{lem:sampling-prob}, the probability of the former is at most $e^{-y^\ast_v}$, whereas the probability of being pruned -- given that $v$ was sampled in the first place -- can be bounded from above by $\Pr[\gamma > y^\ast_v] = 1  - F_\gamma(y^\ast_v)$.
Our choice of $F_\gamma(y) = \frac{1 - \alpha(1 -y)}{1 - e^{-y}}$ guarantees that the expected penalty is just high enough:
\begin{align*}
    \frac{\penaltyprob}{1 - y^\mystar_v} 
    &\leq \frac{(1 - \Pr[v \in T]) + \Pr[v \in T] \cdot (1 - F_\gamma(y^\ast_v))}{1 - y^\mystar_v} \\
    &= \frac{1 - \Pr[v \in T] F_\gamma(y^\ast_v)}{1 - y^\mystar_v}\\
    &\leq \frac{1 - (1 - e^{-y^\ast_y}) F_\gamma(y^\ast_v)}{1 - y^\mystar_v} \\
    &= \alpha.
\end{align*}

Finally, consider a vertex $v$ with $v^\ast_v \in [\sigzero, 1]$.
This time, there is only one case in which we have to pay the penalty for $v$: When $v$ is neither sampled, nor picked up afterwards.
Again, the probability of not being sampled is at most $e^{-y^\ast_v}$ whereas the probability of \emph{not} being picked up  -- given that $v$ was not sampled previously -- is $\Pr[\sigma > y^\ast_v] = 1 - F_\sigma(y^\ast_v)$.
Again, our choice of $F_\sigma(y) = 1 -  \frac{\alpha(1 -y)}{e^{-y}}$ guarantees that 
\begin{align*}
    \frac{\penaltyprob}{1 - y^\mystar_v} 
    &\leq \frac{e^{-y^\ast_v} \cdot (1 - F_\sigma(y^\ast_v))}{1 - y^\mystar_v} = \alpha.
\end{align*}

This concludes our proof of \cref{lem:penalty-bound}.

\subsection{Parity Correction}
In order to analyze the expected cost of both the parity correction step and of our final tour, we first need to further investigate the structure of our computed solution and introduce some additional notation.
Recall that $T$ denotes the union of all sampled trees and that it includes the closed walk $C$, which is obtained by joining all $s_i$-$t_i$-paths $P_i \subseteq T_i$.
Let $R$ be the graph that contains all remaining edges i.e. the (multi-)graph induced by $E(T) \setminus E(C)$.

While it is convenient to think of $C$ as a cycle, note that the paths $P_1, \ldots, P_k$ are not necessarily vertex disjoint. However, we can guarantee that $C$ is a Eulerian multigraph spanning all terminals.
In the same sense, $R$ can be thought of as a collection of small trees which are all rooted at the cycle $C$ (in reality, some of those trees may intersect each other).
In the following, we will call the edges of $C$ \emph{cycle edges} and the edges of $R$ \emph{tree edges}.
A depiction of $C$ and $R$ can be seen in Fig.~\ref{fig:example1-without-cuts}.

When we prune the trees $T_i$ into $T_i'$, the paths $P_i$ are not affected (because the pruning step guarantees that $o_i$ and $o_{i+1}$ stay connected in $T_i^\prime$). 
So our pruning step can only remove edges from $R$.
We thus define $\text{core}(R, \gamma)$ as the graph that is obtained by pruning all trees in $T$ with pruning threshold $\gamma$ and then removing the cycle $C$.
We now partition the edge set of $R$ into layers.
Intuitively, the $i$-th layer of $R$ contains those edges that are contained in $\text{core}(R, \gamma)$ if and only if $\gamma$ does not exceed some value $\eta_i$.

Let $\eta_1 > \ldots > \eta_\ell$ be the $y^\mystar$-values of all vertices that might be affected by our pruning step, i.e.,
$\{\eta_1, \ldots, \eta_\ell\} = \{y_v \mid v \in V \text{ and } \theta \leq y_v \leq \sigma_0\} \cup \{\sigma_0\}$. 
By definition, we always have $\eta_1 = \sigma_0$.
Now let $E_1 = \text{core}(T, \eta_1)$ and $E_i = \text{core}(T, \eta_i) \setminus E_{i-1}$ for $i = 2, \ldots, \ell$.
One can see that $E_1 \cup E_2 \cup \ldots \cup E_\ell$ is a partitioning of $E(R)$.

To be able to bound the cost of the cheapest $\text{odd}(T'')$-join $J$ , we define the following vector
\begin{equation}\label{eq:zvector}
    z \coloneqq \beta x + \underbrace{\sum_{i : \eta_i \geq \gamma} (1 - 2 \eta_i \beta)\chi^{E_i}}_{z_\gamma} + \underbrace{\max\{0, 1 - 2 \beta \sigma\}\chi^{\pickupforest}}_{z_\sigma}
\end{equation}
where $\beta = \frac{1}{3\sigma_0 - \theta}$.
To simplify future arguments, we also define the two components $z_\gamma$ and $z_\sigma$ of $z$, as specified in \eqref{eq:zvector}.
We remark that the vector $\beta x + z_\gamma$ has been used (for a different value of $\beta$) in \cite{blauth-klein-naegele-2024}.
In \cref{subsec:tour-cost}, the cost of $z$ will be used as an upper bound for $c(J)$. To this end, we now prove the following lemma.

\begin{lemma}\label{lem:t-join-poly}
    $z$ lies in the dominant of the $\text{odd}(\solbeforepc)$-join polytope.
\end{lemma}
\begin{proof}
    First, observe that all coefficients $z_e$ are non-negative.
    We now consider an arbitrary $S \subseteq V$ for which $|S \cap \text{odd}(H)|$ is odd and show that $z(\delta(S)) \geq 1$. 

    We begin with the case where $S$ cuts the terminal set,
    i.e., where $0 < |S \cap \{o_1, \ldots, o_k\}| < k$.
    In this case, $S$ separates at least two terminal pairs $(o_{i_1}, o_{i_1+ 1})$ and $(o_{i_2}, o_{i_2 + 1})$, which implies $x(\delta(S)) \geq 2$ and therefore $z(\delta(S)) \geq \beta x(\delta(S)) \geq 1$.
    In the following we can thus assume that $\emptyset \neq S \subseteq V \setminus \{o_1, \ldots, o_k\}$.

    Now suppose that $\delta(S)$ contains a pickup edge $e = \{u, v\} \in E(\pickupforest)$ s.t. $u \in S$ and $v \not \in S$.
    Recall that $\pickupforest$ only spans vertices $w$ with $y_w \geq \sigma \geq \sigma_0$. 
    We therefore have $u \in S \subseteq V \setminus \{o_1, \ldots, o_k\}$ which (by the connectivity constraints of (OLP)) implies $x(\delta(S)) \geq 2 y_u \geq 2 \sigma_0$ and therefore
    $z(\delta(S)) \geq \beta x(\delta(S)) + \max\{0, 1 - 2 \beta \sigma_0\} \geq 1$.

    It remains to show the bound for the case when $\delta(S)$ only contains cycle and tree edges.
    By a simple counting argument, one can see that $|\delta_\solbeforepc(S)|$ has the same parity as $| S \cap \text{odd}(\solbeforepc)|$ and therefore  must be odd. 
    Now observe that by our assumption, $|\delta_{\pickupforest}(S)| = 0$ and that because $C$ is a Eulerian multigraph, $|\delta_C(S)|$ must be even.
    If follows that $\delta(S)$ must have an odd number of tree edges.
    We finish our proof by marginally adapting an argument from \cite{blauth-klein-naegele-2024}:

    First, we consider the case where $\delta(S)$ contains \emph{exactly} one tree edge $e$.
    This is only possible, if $S$ includes a whole subtree $T_e$ of one of the trees in $T^\prime$.
    Because this subtree $T_e$ has survived the pruning step, $e$ must lie in some layer $E_i$ for which $\eta_i \geq \gamma$.
    Furthermore, if $e \in E_i$, then $T_e$ must contain at least one vertex $v$ with $y_v \geq \eta_i$, which implies that $x(\delta(S)) \geq 2 \eta_i$. So we have 
    \[z(\delta(S)) \geq \beta x(\delta(S)) + (1 - 2 \eta_i) \geq 2 \beta \eta_i + (1 - 2 \beta \eta_i) \geq 1.\]
    If, however, $\delta(S)$ contains at least three tree edges, we know that all vertices $v \in S$ have $y_v \geq \theta$ and that each tree edge contributes at least $(1 - 2 \beta \sigma_0)$ to $z_\gamma$. Therefore 
    $z(\delta(S)) \geq 2 \beta \theta + 3 (1 - 2 \beta \sigma_0) \geq 1$.
\end{proof}

\begin{figure}[ht]
    \vskip 0.2in
    \begin{subfigure}{.40\textwidth}

        \centering
        \centerline{\includegraphics[width=\columnwidth]{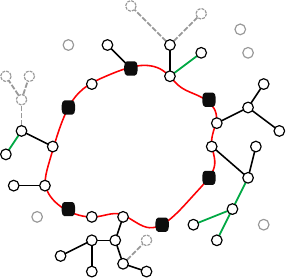}}
        (a)
    \end{subfigure}\hfill
    \begin{subfigure}{.40\textwidth}
        \centering
        \centerline{\includegraphics[width=\columnwidth]{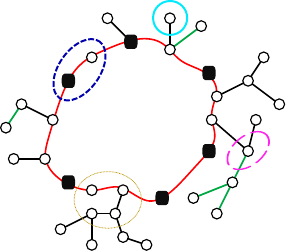}}
        (b) 
    \end{subfigure}
    \caption{
        \label{fig:example1-without-cuts}
        \label{fig:example1-with-cuts}
        \small 
        (a) The graph $T^{\prime\prime}$ after pruning and picking up critical vertices.
        The terminals in $O$ are drawn as black rectangles. The cycle $C$ is depicted in red, the surviving edges of $R$ are drawn in black and edges in $\pickupforest$ in green.
        The greyed out vertices and edges do not belong to $T''$.
        They have  either been pruned (the dashed vertices and edges), not sampled, or split off.
        (b) The same graph $T''$  with the various cuts that are considered in the proof of \cref{lem:t-join-poly} drawn in different colors.
        The dashed blue cut is an example for the case where $0 < |S \cap O| < k$. The dashed pink cut shows the case where a pickup edge (drawn in green) is cut.
        Note that even though the vertex in $S$ was not picked up itself, it is still part of $\pickupforest$ and thus must have a high $y$-value.
        The remaining two cuts (drawn in brown and light blue) show the two cases where $\delta(S)$ contains an odd number of tree edges.
    }
\end{figure}
\subsection{Bounding the Tour Cost}\label{subsec:tour-cost}
In this section, we show the following bound on the cost of our computed tour:
\begin{lemma}
    $\expectation{c(T'')} \leq \alpha \cdot  c(x^\ast)$
\end{lemma}
The total cost of our computed tour can be bounded from above by
\begin{align*}
    \expectation{c(H)} \leq \expectation{c(T^\prime)} + \expectation{c(\pickupforest)} + \expectation{c(J)}.
\end{align*}
Our bound on the cost of $J$ is given by the cost of the parity correction vector $z = \beta x + z_\gamma + z_\sigma$.
We start by analyzing the expected combined value of $c(T^\prime)$ and $c(z_\gamma)$. 
\begin{lemma} \label{lem:cost-of-pruned-tree}
    $\expectation{c(T^\prime) + c(z_\gamma)} \leq c(x) \Bigl(2 + \appfacpcsymbol - \alpha - (2  + 2 \appfacpcsymbol) \beta \sigma_0 + 2 \alpha \beta \sigma_0\Bigr)$
    where $\appfacpcsymbol$ denotes the approximation guarantee of the current best PCTSP algorithm.
\end{lemma}
The main idea of the proof of \cref{lem:cost-of-pruned-tree} is that the weight which each layer $E_i$ is assigned in $z_{\gamma}$ is large when $\eta_i$ is small and vice versa.
At the same time, the chance for layer $E_i$ to be present after pruning and thus to contribute at all to both $z_\gamma$ and $c(T^\prime)$ is an increasing function of $\eta_i$. 
By balancing out these two values, we obtain an upper bound or the contribution of all layers in $R$ to $c(T^\prime) + c(z_\gamma)$.
At the same time we take into account that the cycle $C$ only contributes toward $c(T^\prime)$ but crucially not towards the cost of the parity correction vector.
This is where we utilize our assumption that $\hat C$ and therefore also $C$ can be lower bounded by a constant fraction of $\text{opt}$.
For formally proving \cref{lem:cost-of-pruned-tree}, we need the following technical \cref{lem:technical-proof}.
\begin{lemma}\label{lem:technical-proof}
    For $0< \sigma_0 \le 0.8$, the function $g(y) = \Pr[\gamma \leq y] (2 - 2 \beta y)$ attains its maximum value over the interval $[\theta, \sigma_0]$ at $y = \sigma_0$
\end{lemma}
\begin{proof}
    First, observe that $\Pr[\gamma \leq y] = F_\gamma(y) = \frac{1 - \alpha (1 - y)}{1 - e^{-y}}$ and therefore 
    \[g(y) = 2 F_\gamma(y) (1 - \beta y)\]
    The claim follows from proving that $g$ is monotonously increasing on $[\theta, \sigma_0]$ by showing that its derivative is strictly positive:
    \begin{align}
        \frac 1 2 \frac{d}{dy} g(y) &= f_\gamma(y) (1 - \beta y) - \beta F_\gamma(y)  \nonumber\\
        &\geq f_\gamma(y) (1 - \beta \sigma_0) - \beta  \nonumber\\
        &\geq \alpha (1 - \beta \sigma_0) - \beta \label{ineq:3}\\
        &> 0\label{ineq:4},
    \end{align}
    where $f_\gamma$ denotes the density function of $\gamma$.
    In \eqref{ineq:3} we used that for $y \in [\theta, \sigma_0]$
    \begin{align*}
        f_\gamma(y) 
        &= \frac{\alpha}{1 - e^{-y}} - (1 - \alpha(1-y))\frac{e^{-y} }{(1 - e^{-y})^2} \\
        &= \frac{\alpha}{1 - e^{-y}} - F_\gamma(y) \frac{e^{-y} }{1 - e^{-y}} \\
        &\geq \frac{\alpha - e^{-y}}{1 - e^{-y}} \\
        &= \alpha + \frac{\alpha - 1}{e^y - 1} \\
        &\geq \alpha.
    \end{align*}
    Furthermore, \eqref{ineq:4} follows because $\alpha$ and $\beta$ are determined by $\sigma_0$:
    \[
        \alpha = \frac{e^{-{\sigma_0}}}{1-\sigma_0},\quad
        \theta = 1-\frac{1}{\alpha}, \text{ and}\quad
        \beta = \frac{1}{3\sigma_0 - \theta}.
    \]
    We obtain the function
    \begin{align*}
        h(\sigma_0) &:= \alpha (1 - \beta \sigma_0) - \beta
        = \alpha - \beta(\alpha \sigma_0 +1)\\
        &= \frac{\text{e}^{-\sigma_0}}{1-\sigma_0} - \frac{1}{3 \sigma_0 - 1 + \frac{1}{\frac{e^{-\sigma_0}}{1-\sigma_0}}}\biggl(\frac{\text{e}^{-\sigma_0}}{1-\sigma_0} \sigma_0 + 1\biggr).
    \end{align*}
    The derivative of this function
    \[
    \frac{d \sigma_0}{d t} h(\sigma_0) =
-{\left(\sigma^{2} + 2 \, \sigma - 1\right)} e^{\sigma} - 1
    \]
    is smaller than zero for $\sigma_0 >0$, i.e., the function is monotonously decreasing. Therefore there is only one root in this range and one can easily check that it is located at $\sigma_0 > 0.8$.
\end{proof}
\begin{proof}[Proof of \cref{lem:cost-of-pruned-tree}]
    Our analysis follows the basic approach from \cite{blauth-klein-naegele-2024}, but distinguishes between the cycle $C$ and the remaining part $R$ of the sampled solution $T$ to leverage the lower bound on the the cost of the cycle, which we get by running the simple algorithm from \cref{sec:pcotsp} in parallel.
    Another difference to \cite{blauth-klein-naegele-2024} is that $T$ is not a single tree, but consists of $k$ distinct sampled trees, which requires some additional notation.

    Let $\mathcal T = \mathcal T_1 \times \ldots \times \mathcal T_k$ denote the set of all possible outcomes obtained by sampling the $k$ trees as described in \cref{sec:pcotsp}.
    One can see that the probability of a fixed outcome $\pi = (T_1, \ldots, T_k) \in \mathcal T$ is 
    $\mu(\pi) = \prod\limits_{i=1}^{k} \mu_i(T_i)$
    and that 
    ${\sum \limits_{\pi \in \mathcal T} \mu(\pi) = 1}$.

    Recall that technically, the  subgraphs $C, T$ as well as the layers $E_i$  depend on the combination of sampled trees $\pi$. We therefore write, e.g., $C_\pi$ to refer to the concrete cycle $C$ that arises from sampling the trees $T_i$ in $\pi = (T_1, \ldots, T_k)$. Now
    \begin{align}
        \expectation{c(T^\prime) + c(z_\gamma)} 
        &\leq \expectation{c(C)} + c(R^\prime) + \expectation{c(z_\gamma)} \nonumber\\
        &=
        \sum_{\pi \in \mathcal T} \mu(\pi)\left (  c(C_\pi) + \mathbb E_\gamma\bigl[c(R^\prime_\pi)\bigr] + \sum_{i = 1}^{\ell_\pi} F_\gamma(\eta_j) \left(1 - 2 \beta \eta_j\right)c(E_{i, \pi})\right) \nonumber\\
        &=
        \sum_{\pi \in \mathcal T} \mu(\pi)\left (  c(C_\pi)  +\sum_{i = 1}^{\ell_\pi} \underbrace{\Pr[\gamma \leq \eta_i] \left(2 - 2 \beta \eta_j\right)}_{\eqqcolon g(\eta_i)} c(E_{i, \pi}) \right)\nonumber\\
        &\leq \sum\limits_{\pi \in \mathcal T} \mu(\pi) \Bigl( c(C_\pi) + g(\sigma_0) c(R_\pi)\Bigr) \label{ineq:la8}\\
        &= \sum\limits_{\pi \in \mathcal T} \mu(\pi) \Bigl( c(C_\pi) + g(\sigma_0) \bigl( c(T_\pi) - c(C_\pi) \bigr)\Bigr) \nonumber\\
        &= \sum\limits_{\pi \in \mathcal T} \mu(\pi) \Bigl(  g(\sigma_0) c(T_\pi) - \bigl(f(\sigma_0) - 1\bigr) c(C_\pi) \Bigr) \nonumber\\
        &\leq g(\sigma_0)c(x) - \bigl(g(\sigma_0) - 1\bigr) c(\hat C) \nonumber\\
        &= 2\bigl(1 - \beta \sigma_0\bigr) c(x) - \bigl(1 - 2 \beta \sigma_0\bigr) c(\hat C) \nonumber
    \end{align}
    In \eqref{ineq:la8} we used the fact from \cref{lem:technical-proof} that $g(\eta_i)$ is maximized at $\eta_i = \sigma_0$ and that the layers $E_{i, \pi}$ partition the edge set of $R_\pi$. 
    Now recall that we may assume that $c(\hat C) \geq (\alpha - \appfacpcsymbol) \text{opt}_{\PCOTSP} \geq (\alpha - \appfacpcsymbol) c(x)$, which yields a bound of 
    \begin{align*}
        \expectation{c(T^\prime) + c(z_\gamma)} 
        &\leq 2\bigl(1 - \beta \sigma_0\bigr) c(x) - \bigl(1 - 2 \beta \sigma_0\bigr) c(\hat C) \\
        &{\leq} c(x) \Bigl(2 - 2 \beta \sigma_0 - (1 - 2 \beta \sigma_0) (\alpha - \appfacpcsymbol)\Bigr) \\
        &= c(x) \Bigl(2 - 2 \beta \sigma_0 - \alpha + \appfacpcsymbol + 2 \alpha \beta \sigma_0 - 2 \appfacpcsymbol \beta \sigma_0\Bigr) \\
        &= c(x) \Bigl(2 + \appfacpcsymbol - \alpha - (2  + 2 \appfacpcsymbol) \beta \sigma_0 + 2 \alpha \beta \sigma_0\Bigr)
    \end{align*}
\end{proof}

We will next analyze the cost of $\pickupforest$, in \cref{lem:pickup_cost}.  
The proof builds on the following theorem.

\begin{theorem}[Böhm et al.~\cite{boehm-friggstad-moemke-2025}]\label{thm:pickup_cost}
    Let $X \subseteq U \subseteq V$. Furthermore, let $S \subseteq U \setminus X$ be a randomly chosen subset such that $\Pr[v \not \in S] \leq \rho$ for each $v \in X$. Let $F_{X}$ and $F_{X \cup S}$ denote the the cheapest $X$-rooted spanning forest and the cheapest $(X \cup S)$-rooted spanning forest of $U$ respectively. Then $\expectation{c(F_{X \cup S})} \leq \rho \cdot c(F_X)$.
\end{theorem}

\begin{lemma}\label{lem:pickup_cost}
    For a fixed value of $\sigma \in [\sigma_0, 1]$, the cost of $\pickupforest$ is at most $\frac{e^{-\sigma}}{\sigma} \cdot c(x)$.
\end{lemma}
\begin{proof}
    We invoke \cref{thm:pickup_cost} with $U = V_\sigma$, $X = O$ and $S = (V_{\sigma} - O) \cap V(T)$. 
    Recall that our pickup forest $\pickupforest$ is the cheapest $V_\sigma \cap V(T)$ rooted spanning forest of $V_\sigma$ and observe that $V_\sigma \cap V(T) = X \cup S$.
    By \cref{lem:sampling-prob}, each vertex $v \in U \setminus X$ is $\emph{not}$ in sampled (and therefore not in $S$) with probability at most $e^{-\sigma}$.
    It remains to show that the cost of the cheapest $O$-rooted spanning forest of $V_\sigma$ is at most $\frac{c(x)}{\sigma}$.

    As we have observed already in \cref{sec:pcotsp}, the cost of the cheapest $O$-rooted spanning forest of $V_\sigma$ is equal to the cost of a minimum spanning tree in the graph obtained by contracting all vertices of $O$ into a single vertex $r$, i.e., on $V'' = (V_\sigma - O) \cup \{r\}$.
    Furthermore, we have also observed that by applying the very same contraction to our solution ($x, y$), we obtain a feasible solution $(x^\prime, y^\prime)$ to the (non-ordered) PCTSP relaxation.
    By splitting off all vertices in $V - V_\sigma$ and scaling up by a factor of $\frac{1}{\sigma}$, we obtain a vector $x''$ for which $x^\prime(\delta(s)) \geq 2$ for all $\emptyset \neq S \subseteq V''$, i.e., a feasible point in the dominant\footnote{It is possible to obtain a feasible point in the polytope itself by applying a sequence of splitting-off operations.} of the subtours elimination polytope of $V''$.

    The cost of the minimum spanning tree on $V''$ is therefore at most $c(x'') \leq \frac{1}{\sigma} c(x)$, which concludes the proof.
\end{proof}

Equipped with this upper bound, we can now bound the \emph{expected} cost of $\pickupforest$, utilizing the density function $f_\sigma(y) = \frac{d}{dy} F_\sigma(y) = \alpha y e^y$ and integrating over the range $[\sigma_0, 1]$ from which $\sigma$ is drawn:
\begin{align}\label{eqn:expected-pickup-cost}
    \mathbb E\bigl[c(\pickupforest)\bigr]
    = c(x)\int_{\sigma_0}^{1} f_\sigma(y) \frac{e^{-y}}{y} dy 
    = \alpha c(x) \int_{\sigma_0}^{1} 1 dy 
    = \alpha (1 -\sigma_0) c(x)
    = e^{-\sigma_0} c(x)
\end{align}
For the last equality we have used the definition of $\sigma_0$.
We emphasize that by randomizing the choice of $\sigma$ instead of flatly using $\sigma = \sigma_0$, we have gained a factor of $\sigma_0$. 
In a similar way, we can use \cref{lem:pickup_cost} to compute the expected cost of $z_\sigma$.

\begin{lemma} \label{lem:cost-of-z-sigma}
    $\expectation{c(z_\sigma)} = \alpha c(x) \left(\frac{1}{4 \beta} - \sigma_0  + \beta \sigma_0^2\right)$
\end{lemma}
\begin{proof}
Note that when $\sigma > \frac{1}{2\beta}$, then $z_\sigma$ is by definition the zero-vector. We compute
\begin{align*}
    \mathbb E\bigl[c(z_\sigma)\bigr]
    &= \mathbb E\bigl[\max\{0, 1 - 2 \beta \sigma\}c(\chi^{\pickupforest}) \bigr] \\
    &= c(x) \int_{\sigma_0}^{\frac{1}{2 \beta}} f_\sigma(y) \frac{e^{-y} (1 - 2 \beta y)}{y} dy \\
    \nonumber&= \alpha c(x) \int_{\sigma_0}^{\frac{1}{2 \beta}} (1 - 2\beta y) dy \\ 
     &= \alpha c(x) \left(\frac{1}{2 \beta} - \sigma_0 - 2 \beta \int_{\sigma_0}^{\frac{1}{2 \beta}} y \ dy \right)\\
     \nonumber&= \alpha c(x) \left(\frac{1}{2 \beta} - \sigma_0 - 2 \beta \left(\frac{1}{8 \beta^2} - \frac{\sigma_0^2}{2}\right) \right) \\ 
     \nonumber&= \alpha c(x) \left(\frac{1}{4 \beta} - \sigma_0  + \beta \sigma_0^2\right)
\end{align*}

\end{proof}
Finally, we are able to combine the bounds on the various parts of $c(H)$, and obtain our final upper bound on the expected tour cost:
\begin{align*}
    \expectation{c(H)} &\leq  \expectation{c(T^\prime)} + \expectation{c(\pickupforest)} + \expectation{c(J)} \\
    &= \expectation{c(T^\prime) + c(z_\gamma)} + \expectation{c(\pickupforest)} + \expectation{c(z_\sigma)} + \beta c(x^\ast) \\
    &\leq c(x^\ast) \Bigl( 2 + \appfacpcsymbol + \beta - \alpha - (2 + 2\appfacpcsymbol) \beta \sigma_0 + 2 \alpha \beta \sigma_0 + e^{-\sigma_0} + \frac{\alpha}{4\beta} -  \alpha \sigma_0 + \alpha \beta\sigma_0^2 \Bigr)
\end{align*}
Note that the values of $\beta$ and $\sigma_0$ are functions of $\alpha$. We therefore can express our upper bound as $\expectation{c(H)}~\leq~c(x^\ast) f(\alpha, \appfacpcsymbol)$.
By \cref{lem:penalty-bound}, \cref{alg:pcotsp} pays at most $\alpha$ times the penalty incurred by $(x^\ast, y^\ast)$.
Running \cref{alg:pcotsp} with $\theta$ and $\sigma$ determined by the single parameter $\alpha$ thus yields an approximation factor of $\max(\alpha, f(\alpha, \appfacpcsymbol))$.
Recall that the value of $\appfacpcsymbol$ is currently slightly below $1.599$.
For $\appfacpcsymbol = 1.599$, the term  $\max(\alpha, f(\alpha, \appfacpcsymbol))$ is minimized at $\alpha \approx 2.096896 < \appfac$.
In fact, if we set $\alpha = \appfac$, the term evaluates to $\appfac$.

For $\alpha = \appfac$ and $\appfacpcsymbol = 1.599$, we have $\beta \approx 0.548775$ and $\sigma_0 \approx 0.781790$.
Thus we have proven \cref{thm:opctsp}.

\section{Prize-collecting \kpath TSP}
\label{sec:ktsp}
\PCkTSP can be modeled as the following linear program (kLP). 
\begin{mini*}
    {}{\displaystyle\sum\limits_{e \in E} \displaystyle\sum\limits_{i = 1}^k c_e x_{i, e} + \displaystyle\sum\limits_{v \in V-\mathcal T} \pi_v (1 - y_v)}{\label{lp:kpath}}{}
    \addConstraint{y_v}{= \sum\limits_{i = 1}^k y_{i, v}\quad }{\forall v \in V}
    \addConstraint{(x_i, y_i) \text{ lies in the }s_i\text{-}t_i\text{-stroll relaxation}}{}{}
\end{mini*}

Similar to \cite{boehm-friggstad-moemke-2025}, we describe two algorithms for the problem 
and show that an appropriate combination of the two algorithms gives an \appfackpath-approximation. 

In one algorithm, (which we call Algorithm $A$), we first sample $k$ trees using \cref{lem:bang-jensen} and an optimal solution $(x^\ast, y^\ast)$ of (kLP), 
where tree $T_i$ connects terminals $s_i, t_i$.
The remaining vertices are picked up in a probabilistic fashion akin to Algorithm \ref{alg:pcotsp},
i.e., we choose a random threshold $\sigma \in [\sigma_0', 1]$ 
and buy a $\bigl(V_\sigma \cap \bigcup_{i=1}^k V(T_i)\bigr)$-rooted spanning forest $\pickupforest$ of $V_\sigma$.
Here, $\sigzeroprime$ is a constant whose value will be determined later.  
Then we double every edge that does not, for any $i$,  lie on the $s_i$-$t_i$ path in $T_i$. 
This gives a tour $H_A$. 
Define $\Delta := \sum c(s_i,t_i)$ and $\eta \coloneqq \frac \Delta {c(x^\ast)}$.
Then it is easy to see that 
\begin{align*}
    c(H_A) 
    &\le 2 c(x^\ast) +2 \expectation{c(\pickupforest)} - \Delta
\end{align*}
One can already see that intuitively, this yields good results whenever $\Delta$ is large.

Now we define the distribution from which $\sigma$ is drawn.
We choose $\sigma$ s.t. 
$\Pr[\sigma \leq y] = F'_\sigma(y)$ where
\[F'_\sigma(y) = 1 - \frac{e^{y} (1 - y)}{e^{ \sigzeroprime} (1 - \sigzeroprime)}.\]
Note that except for the constant $\sigzeroprime$, this is exactly the same distribution that we used in \cref{alg:pcotsp}.
In fact, if we define $\rho \coloneqq \frac{e^{-\sigzeroprime}}{1-\sigzeroprime}$, we obtain $F'_\sigma(y) = 1 - \frac{\rho (1-y)}{e^{-y}}$  (compare this to $F_\gamma(y)= 1 - \frac{\alpha (1-y)}{e^{-y}}$) so any result about $F_\sigma$ obtained in the previous setting carries over to $F'_\sigma$ if we replace $\alpha$ by $\rho$.
We remark that we use $\rho$ instead of $\alpha'$ because unlike in the previous case, $\rho$ will not be our final approximation factor.

We can thus bound the expected   cost of $\pickupforest$ in the same way as we did for \PCOTSP. 
This is possible because the distribution $F_\sigma$ as well as the (lower bound on the) probability of a vertex $v \in V_\sigma$ sampled at least once are the same as in \cref{subsec:tour-cost}.

First, we prove an equivalent of \cref{lem:pickup_cost} i.e. we show that $c(\pickupforest) = \frac{e^{-\sigma}}{\sigma}$ for any fixed value of $\sigma$,
and then we compute the expected cost $\expectation{c(\pickupforest)} = e^{-\sigzeroprime}$ by integrating over the range $[\sigzeroprime, 1]$ from which $\sigma$ is chosen as we have done in \cref{eqn:expected-pickup-cost}.
This gives us the following upper bound on the expected tour cost:
\begin{align*}
    c(H_A) 
    &\le \biggl(2+2 e^{-\sigzeroprime} - \eta\biggr) c(x^\ast).
\end{align*}

Furthermore, by similar reasoning as in \cref{sec:penalty-cost}, 
we know that the expected total penalty paid by this algorithm is at most $\rho  = \frac{e^{-\sigzeroprime}}{1 - \sigzeroprime}$ times the fractional penalty cost incurred by $(x^\ast, y^\ast)$.

Now we give a simple Algorithm $B$ that works well when $\eta$ is small. 
Contract all the $2k$ terminals into one mega-vertex $w$ with penalty $\infty$, 
solve the \PCTSP LP for this instance, and sample a tree $T$ from the solution. 
Double every edge in $T$ (obtained in the original graph), and then add the $k$ edges $\{s_i,t_i\}$ 
to get a final solution $H_B$. 
It is easy to see that $c(H_B)\le (2+\eta)c(x^\ast)$ and that the incurred penalty is no higher than the fractional penalty cost of $(x^\ast, y^\ast)$.

Running both algorithms $A$ and $B$ and returning each solution with probability $\frac{1}{2}$, yields a tour of expected cost
$\frac{c(H_A) + c(H_B)}{2} \leq (2 + e^{-\sigma_0'}) \cdot c(x^\mystar)$
and an expected total penalty cost of at most
\[ \frac{1}{2}\biggl( \frac{e^{-\sigma_0'}}{1 - \sigma_0'} + 1\biggr) \cdot \sum_{v \in V - \mathcal{T}}(1 - y^\ast_v).
\]
The approximation ratio that we get from combining both algorithms is 
\[
    \max\biggl\{
        2 + e^{-\sigma_0'},
        \frac{1}{2}\biggl( \frac{e^{-\sigma_0'}}{1 - \sigma_0'} + 1\biggr)
    \biggr\},
\]
which is minimized at $\sigma_0 \approx 0.892769$. This yields an approximation guarantee slightly below $2.41$, 
proving \cref{thm:kpathtsp}.

\section{Discussion}
The \PCOTSP, as a generalization of both \PCTSP and \OTSP, 
poses the challenges of each of the individual problems plus new challenges. 
The best approximation algorithms for \OTSP (\cite{boehm-friggstad-moemke-2025,armbruster-mnich-nagele-2024}) 
both pick up vertices which have been left out of sampling, which imposes a cost of $\frac{1}{e}$ over the solution. 
In the latest \PCTSP result (\cite{blauth-naegele-2023}), the cost of parity correction is slightly below $0.6$. 
Hence, even without considering the specific characteristics of \PCOTSP, 
utilizing the techniques of the latest results for the two problems brings the approximation factor close to 2. 
But note that in \PCOTSP, the picking up of vertices is more complicated because the vertices are not necessarily wholly covered fractionally, 
and the need for parity correction for picked up vertices adds further complications which are absent from \PCTSP. 

It is not difficult to see that for the special cases of \PCTSP, \OTSP, our algorithms produce 
the best previously-known approximation factors for these problems. 
For example, when $k=1$, \PCOTSP is simply \PCTSP. 
In this case, the cycle length over the one vertex is zero, 
and the output of our algorithm would be no worse than 
the output \PCTSP algorithm of \cite{blauth-klein-naegele-2024} on the remaining vertices.
Likewise, setting all penalties to $\infty$ turns \PCOTSP into \OTSP (and thus all $y$ values are 1). 
Thus all vertices that have not been sampled will be picked up (i.e., $\sigma=1$), 
and no vertex is split off. This is equivalent to the algorithms of 
\cite{boehm-friggstad-moemke-2025} and \cite{armbruster-mnich-nagele-2024}. 
In a similar vein, setting all penalties to $\infty$ turns \PCkTSP into \kTSP, and here 
our algorithm would do the same as the factor $2.367$ algorithm of \cite{boehm-friggstad-moemke-2025}. Their factor $2.2131$ algorithm, however, does not directly carry over to our setting. The issue is that in the prize collecting setting, we require a picking-up step which leads to additional costs exceeding the additional gains as soon as we have to sample more than one tree.

It is an intriguing open question whether the problem admits an approximation factor of 2 or below, 
which we believe requires improvements of at least one of \PCTSP or \OTSP. 
Likewise, finding a good lower bound on the integrality gap of current LP relaxation would be very interesting. 

\bibliography{references}
\end{document}